\documentclass[11pt, a4paper]{article}
\usepackage{fullpage}
\usepackage{amsmath,amssymb,amsthm,mathtools,cite,comment,graphicx}
\usepackage{hyperref}

\usepackage{tikz}
\usetikzlibrary{arrows.meta, positioning,calc}

\newcommand{\subseteqr}{\subseteq_\mathrm{r}}
\newcommand{\Pre}{\mathrm{Pre}}
\newcommand{\Suf}{\mathrm{Suf}}

\theoremstyle{plain}
\newtheorem{theorem}{Theorem}[section]
\newtheorem{lemma}[theorem]{Lemma}

\theoremstyle{definition}
\newtheorem{definition}[theorem]{Definition}

\theoremstyle{remark}

\title{Approximate Maintenance of Maximum Subarray Sum\\
in the Sliding Window Model}
\author{
Ryo Suzuki$^*$
\and
Yutaro Yamaguchi\thanks{Graduate School of Information Science and Technology, Osaka University, Japan. Email: \texttt{yutaro.yamaguchi@ist.osaka-u.ac.jp}}
}
\date{\empty}

\begin{document}

\maketitle
\thispagestyle{empty}

\begin{abstract}
In the sliding window model, we are required to maintain the target statistics over the most recent $n$ elements of a data stream, which is captured by a window of size $n$ sliding over the data stream.
Exact computation usually requires space linear in $n$, and the central goal is approximate maintenance using sublinear space.
In this paper, we study the problem of maintaining the maximum subarray sum in the sliding window model. 
While the classical Kadane's algorithm computes the exact answer using constant space in the static setting, it does not extend directly, because a new element makes the oldest one expire, which may invalidate the optimal subarray so far.
Our first observation is that the so-called Smooth Histogram framework can lead to a constant-factor approximation (in the sense of relative error) using $O((\log n)^2)$ bits of space.
We then refine this framework accordingly, which enables for any $\epsilon > 0$ to maintain a $(1 \pm \epsilon)$-approximation using $O(\epsilon^{-1}(\log n)^2)$ bits of space and $O(\epsilon^{-1}\log n)$ operations per update.
The space complexity is asymptotically optimal.

\end{abstract}

\newpage
\pagenumbering{roman}
\tableofcontents
\newpage
\pagenumbering{arabic}
\setcounter{page}{1}

\section{Introduction}
The rapid growth of sequential data such as network logs, sensor measurements, and financial transactions has led to increasing interest in data stream processing.
In many applications, only the most recent data is relevant for decision making, which motivates the \emph{sliding window model}.
In this model, at each time $t$, computations are performed on the most recent $n$ elements, while older elements expire automatically.

Maintaining statistics (e.g., mean, variance, and counting specific elements) is a typical task in the sliding window model.
It is easy if we have plenty of time and memory, but becomes challenging under strict time or memory constraints.
Storing all elements requires $\Omega(n)$ space, which is impractical for large-scale streams.
Therefore, the central goal is to maintain an \emph{approximation} of the target statistics using sublinear space.
For many problems in real-time applications, the time complexity of each update is also significant.

\subsection{Sliding Window Model}
We formalize the sliding window model and the goal.
Let $(p_1, p_2, \dots)$ be a \emph{stream} of data, and $n$ be the \emph{window size}.
As a usual setting, we assume that all elements $p_i$ are integers whose absolute values are bounded by $O(\mathrm{poly}(n))$.
At each time $t = 1, 2, \dots$, the \emph{active window} is defined as
\begin{align}
W_t \coloneqq
\begin{cases}
(p_1, p_2, \dots, p_t) & (t < n), \\
(p_{t-n+1}, p_{t-n+2}, \dots, p_t) & (t \ge n).
\end{cases}\label{eq:active_window}
\end{align}
Let $f(W_t)$ denote the true value of the target statistics of $W_t$.
Then, for $\epsilon > 0$, an estimation $f'(W_t)$ of $f(W_t)$ is called a \emph{$(1 \pm \epsilon)$-approximation} if
\begin{align}
(1 - \epsilon)f(W_t) \le f'(W_t) \le (1 + \epsilon)f(W_t).\label{eq:approximation}
\end{align}
An algorithm is also said to be \emph{$(1 \pm \epsilon)$-approximation} if at every time $t$ it maintains a $(1 \pm \epsilon)$-approximation of $f(W_t)$.

The sliding window model was formally introduced by Datar, Gionis, Indyk, and Motwani~\cite{Datar2002}.
They proposed the \emph{Exponential Histogram (EH)} technique to maintain a certain type of statistics.
As the simplest example, it shows that for any $\epsilon > 0$, a $(1 \pm \epsilon)$-approximation of counting $1$s in the sliding window over a bit stream (where every $p_i$ is $0$ or $1$) can be maintained only using $O(\epsilon^{-1}(\log n)^2)$ bits of space (see Section~\ref{sec:EH} for more details).
This work initiated a large body of research on sliding window algorithms (see, e.g., \cite{braverman2016sliding, muthukrishnan2005data}).

While the EH technique can handle \emph{weakly additive} functions efficiently, many important statistics fall outside this class.
To address this limitation, Braverman and Ostrovsky~\cite{BravermanOstrovsky2007} introduced the \emph{Smooth Histogram (SH)} framework, which extends the applicability of sliding window algorithms to a broader class of functions satisfying a \emph{smoothness} property (see Section~\ref{sec:SH} for the details).

Subsequent work has further explored both the power and limitations of this framework.
In particular, Krauthgamer and Reitblat~\cite{Krauthgamer2022} introduced the notion of \emph{almost-smooth} functions, showing that the smooth histogram paradigm can be extended to functions that do not strictly satisfy smoothness, by relaxing the structural conditions.

In addition, several works have improved space bounds for specific problems within the sliding window model.
For example, Braverman, Grigorescu, Lang, Woodruff, and Zhou~\cite{Braverman2018} established nearly optimal bounds for counting distinct elements and heavy hitters, showing that polylogarithmic space is both necessary and sufficient up to lower-order factors.
More recently, Feng, Swartworth, Woodruff~\cite{Feng2025} further tightened these bounds by introducing the notion of \emph{strong estimators}, which allow more flexible error guarantees while preserving high-probability correctness.

Despite these advances, there remain fundamental limitations.
Certain functions, such as inversion counting, exhibit strong global dependencies and do not satisfy smoothness-type properties.
Indeed, small differences in the prefix discarded to save memory may lead to arbitrarily large deviations after future updates, making such functions unsuitable for SH-based approaches (see also~\cite{BravermanOstrovsky2007, Krauthgamer2022}).

In parallel, alternative models and distributed variants of sliding window algorithms have also been studied~\cite{Gibbons2002}.

\subsection{Maximum Subarray Sum}
In this paper, we study the problem of maintaining the \emph{maximum subarray sum} in the sliding window model.
Given a sequence $W = (x_1, x_2, \dots, x_n)$, the maximum subarray sum is defined as
\begin{align}
f(W) \coloneqq \max\left\{0, \max_{1 \le i \le j \le n} \sum_{k=i}^{j} x_k\right\}.\label{eq:MSS}
\end{align}
Note that when all the elements in $W$ are nonpositive, so are the summations in \eqref{eq:MSS}, which is the only case that $f(W) = 0$.
Besides, if $x_i \in \{0, 1\}$ for all $i$, then $f(W)$ exactly counts 1s in $W$.
Thus, if we restrict ourselves to the sliding window over a bit stream, this problem coincides with counting 1s.

This problem has a long history and, according to Bentley~\cite{Bentley1984}, originated in earlier work of Grenander on pattern analysis, first in a two-dimensional setting.
It is also fundamental in the context of time series data, which captures local bursts, trends, and anomalies.

In the static setting under the standard word RAM model, the problem admits a linear-time (moreover, constant-time per element) and constant-space solution via Kadane's algorithm~\cite{Bentley1984} (see Section~\ref{sec:Kadane} for the details).
However, this approach cannot be directly extended to the sliding window model, since the expiration of elements may invalidate previously optimal subarrays.
A naive solution requires $\Theta(n)$ time per update and $\Theta(n)$ space, which is impractical.

\subsection{Our Contribution}
In this paper, we prove the following:

\begin{theorem}\label{thm:main}
    For any $\epsilon > 0$, a $(1 \pm \epsilon)$-approximation of the maximum subarray sum in the sliding window can be maintained using $O(\epsilon^{-1}(\log n)^2)$ bits of space and $O(\epsilon^{-1} \log n)$ operations per update.
\end{theorem}

Our first observation is that the maximum subarray sum satisfies a relaxed smoothness condition.
This observation enables a $(1 \pm \alpha)$-approximation using $O((\log n)^2)$ bits of space for a constant $\alpha > \frac{1}{2}$ as a direct application of the standard SH framework.

In order to achieve $(1 \pm \epsilon)$-approximation for any $\epsilon > 0$, we need to refine it accordingly.
In particular, by checking not only the maximum subarray sum but also the maximum suffix sum when pruning maintained information, we can preserve the possibility of future suffix extension.
This can be done without essentially increasing space complexity, which leads to Theorem~\ref{thm:main}.

Note that the space complexity is asymptotically optimal.
Indeed, the lower bound of $\Omega(\epsilon^{-1}(\log n)^2)$ bits for small $\epsilon > 0$ was shown in \cite{Datar2002} for counting 1s over bit streams, which is a special case of our problem as mentioned above.

\subsection{Organization}
The rest of the paper is organized as follows.
In Section~\ref{sec:preliminaries}, we review basics on the Exponential Histogram (EH) technique for the sliding window model and Kadane's algorithm for the maximum subarray sum problem.
In Section~\ref{sec:SH}, as the key ingredient, we explain the idea of the Smooth Histogram (SH) framework beyond the EH technique.
In Section~\ref{sec:main}, we prove Theorem~\ref{thm:main} as the main result of this paper.
In Section~\ref{sec:conclusion}, we conclude the paper with several possible directions of future work.

\section{Preliminaries}\label{sec:preliminaries}
\subsection{Exponential Histogram (EH)}\label{sec:EH}
We briefly review the Exponential Histogram (EH) proposed in \cite{Datar2002}, which is a fundamental technique to maintain approximate statistics in the sliding window.
This technique is applied to a class of functions known as \emph{weakly additive} functions.
Intuitively, such functions can be approximately decomposed over disjoint intervals.

\begin{definition}\label{def:weakly_additive}
A function $f$ defined over intervals of the stream is \emph{weakly additive}\footnote{This is the definition for \emph{superadditive} functions. It can be analogously defined for \emph{subadditive} functions by modifying the second condition: take $0 < C_f \le 1$ and flip the inequalities in \eqref{eq:weakly_additive}.} if it satisfies the following:
\begin{itemize}
\item $0 \le f(A) = O(\mathrm{poly}(n))$ for any interval $A$;
\item there exists a constant $C_f \ge 1$ such that for any pair of adjacent intervals $A$ and $B$,
\begin{align}
\hspace{-7mm}f(A) + f(B) \le f(A \cup B) \le C_f (f(A) + f(B));\label{eq:weakly_additive}
\end{align}
\item $f$ admits a \emph{sketch} (giving an approximation) such that for adjacent intervals $A$ and $B$, a sketch for $f(A \cup B)$ is efficiently computed from sketches for $f(A)$ and $f(B)$.
\end{itemize}
In particular, if we can take $C_f = 1$, we say $f$ is \emph{additive}; then, \eqref{eq:weakly_additive} is always satisfied with equality, and the exact value of $f$ is eligible for its sketch (without any error).
\end{definition}

The idea of the Exponential Histogram is sketched as follows.
For the sake of simplicity, we consider the problem of counting $1$s over a bit stream; then, $f$ is clearly additive.
The algorithm maintains a sequence of \emph{buckets}, each representing an interval of the stream.
Each bucket maintains the function value $f(A)$ of the corresponding interval $A$ and also a \emph{timestamp}, which represents when the bucket can be deleted (i.e., when all $1$s contributing to $f(A)$ in the bucket are completely outside the sliding window).
The values maintained in buckets grow exponentially $(1, 2, 4, \dots)$, and the number of buckets of each value is bounded by a prespecified parameter, which controls the trade-off between the relative error of estimation and the space complexity.

More specifically, when a new element arrives,
\begin{itemize}
\item if it is $1$, a new bucket consisting of the single element is created, and
\item if there are too many buckets of the same size, the oldest two among them are merged into a single larger bucket.
\end{itemize}
In addition, if the oldest element that has expired instead of the new arrival matches the timestamp of the oldest bucket, it is deleted.

At each time, the function value is estimated by aggregating the values maintained in the buckets.
As $f$ is additive, the only error comes from the oldest bucket, which may partially overlap with the window.
For any $\epsilon > 0$, it can be bounded by $\epsilon$ in the relative sense by appropriately setting the upper bound $k_\epsilon = \Theta(\epsilon^{-1})$ of the number of buckets of each value.
Then, the number of buckets is always bounded by $O(\epsilon^{-1} \log n)$, which concludes that $O(\epsilon^{-1}(\log n)^2)$ bits of space is sufficient (as $f(A) \le n$ here, and $f(A) = O(\mathrm{poly}(n))$ in general).

This idea is naturally extended to general additive functions.
Moreover, if $f$ is weakly additive, a similar trade-off can be achieved by appropriate adjustment.
Intuitively, there are other types of estimation error caused by nonadditivity of $f$ and by inaccuracy of the sketch of $f$, which can be bounded using $C_f$ and another parameter representing inaccuracy of the sketch, respectively.
See the original paper \cite{Datar2002} for the details.

\subsection{Kadane's Algorithm for Maximum Subarray Sum}\label{sec:Kadane}
We review Kadane's algorithm~\cite{Bentley1984} for computing the maximum subarray sum in the static setting.
Let $(x_1, x_2, \dots, x_n)$ be the input sequence of integers.
For $i = 0, 1, 2, \dots, n$, the algorithm computes the maximum subarray sum of the prefix $(x_1, \dots, x_i)$ by a standard dynamic programming technique.

At each point, the algorithm maintains the following two pieces of information, each requiring $O(1)$ words (and hence $O(\log n)$ bits under the assumption that $|x_i| = O(\mathrm{poly}(n))$):
\begin{itemize}
\item the maximum $\mu_i$ of the sum of a subarray ending at position $i$ if it is positive, or $\mu_i \coloneqq 0$ otherwise, and
\item the global maximum $\mu^*_i$ of $\mu_j$ over all $j \le i$.
\end{itemize}
They are initialized as $\mu_0 = \mu^*_0 = 0$, and updated for each $i < n$ as follows:
\begin{itemize}
\item $\mu_{i+1} \leftarrow \max\{0, \mu_i + x_{i+1}\}$, and
\item $\mu^*_{i+1} \leftarrow \max\{\mu^*_i, \mu_{i+1}\}$.
\end{itemize}
This yields a simple $O(n)$-time algorithm.
Since $\mu_i$ and $\mu^*_i$ are useless after obtaining $\mu_{i+1}$ and $\mu^*_{i+1}$, it can be implemented using $O(1)$ words of space in total.
It is easy to output a subarray itself attaining the maximum by storing the starting and ending points of the corresponding subarrays in addition to the values of $\mu$ and $\mu^*$.

In the sliding window model, the situation is significantly more challenging.
When the window moves,
\begin{itemize}
\item a new element is inserted at the end, and
\item the oldest element is deleted.
\end{itemize}
The insertion at the end is easily handled in the same way, but the deletion may invalidate the optimal subarray so far.
A naive approach requires restoring the elements in the window and recomputing the answer from scratch, leading to $\Theta(n)$ time per update and $\Theta(n)$ words of space.

\section{Smooth Histogram (SH)}\label{sec:SH}
We explain the idea of the Smooth Histogram (SH) framework introduced by Braverman and Ostrovsky~\cite{BravermanOstrovsky2007}, which provides a general method for approximating functions over sliding windows beyond weakly additive functions.

As discussed in the previous section, the EH technique relies on weak additivity and fails for functions whose values highly depend on interactions across interval boundaries.
The SH framework overcomes this limitation by maintaining multiple overlapping summaries of the stream.

The framework applies to functions satisfying a structural property called \emph{smoothness}.
For two intervals $A$ and $B$, we mean by $B \subseteqr A$ that $B$ is a \emph{suffix} of $A$, i.e., $B$ is included in $A$ and their rightmost elements are the same.

\begin{definition}\label{def:smooth}
For $0 < \beta \le \alpha < 1$, a function $f$ defined over intervals of the stream is \emph{$(\alpha, \beta)$-smooth}\footnote{This is the definition for monotonically \emph{nondecreasing} functions. As mentioned in Section~\ref{sec:nonempty}, it can be analogously defined for monotonically \emph{nonincreasing} functions by modifying the second condition accordingly.} if it satisfies the following:
\begin{itemize}
\item $0 \le f(A) = O(\mathrm{poly}(n))$ for any interval $A$;
\item for any pair of intervals $A$ and $B$ with $B \subseteqr A$,
\begin{itemize}\vspace{-1.5mm}
    \item $f(B) \le f(A)$, and
    \item if $(1-\beta)f(A) \le f(B)$, then $(1-\alpha)f(A \cup C) \le f(B \cup C)$ for any interval $C$ adjacent to $A$ and $B$.
\end{itemize}
\end{itemize}
\end{definition}

The idea of the Smooth Histogram is sketched as follows.
At each time $t$, it maintains a sequence of indices
\[
s_1 < s_2 < \cdots < s_q = t,
\]
each corresponding to an instance for which a base algorithm computes $f(I_i)$ exactly over the interval $I_i = [s_i, t]$, which starts at position $s_i$ and ends at position $t$.
Intuitively, these indices represent candidate starting points of the window.
The key invariant is that the function values between consecutive indices do not differ significantly (unless the starting points are consecutive).

More specifically, when a new element arrives,
\begin{itemize}
\item all maintained instances are updated by adding the element at the end,
\item a new instance starting at the current time is created, and
\item redundant instances are removed so that $p_{s_1}$ is expired but $p_{s_2}$ is still active, and at least one of the following conditions is satisfied for any index $s_i$ $(1 \le i < q)$:
\begin{itemize}\vspace{-1.5mm}
    \item $s_{i+1} = s_i + 1$ and $f(I_{i+1}) < (1 - \beta)f(I_i)$;
    \item $f(I_{i+1}) \ge (1 - \alpha)f(I_i)$ and if $i + 2 \le q$ then $f(I_{i+2}) < (1 - \beta)f(I_i)$.
\end{itemize}
\end{itemize}
The last condition implies $q = O(\beta^{-1}\log n)$ since
\[0 \le f(I_q) \le \dots \le f(I_2) \le f(I_1) = O(\mathrm{poly}(n)).\]
Moreover, the $(\alpha, \beta)$-smoothness of $f$ enables us to maintain a $(1 \pm \alpha)$-approximation of $f(W_t)$, which is achieved by $f(I_2)$.
Indeed, $I_2 \subseteqr W_t \subseteqr I_1$ implies
\[(1 - \alpha)f(I_1) \le f(I_2) \le f(W_t) \le f(I_1),\]
unless $I_2 = W_t$ (in this case, obviously $f(I_2) = f(W_t)$).
Thus, this always gives an approximation with one-sided relative error, i.e., \[(1 - \alpha)f(W_t) \le f(I_2) \le f(W_t).\]

\begin{theorem}[Braverman and Ostrovsky~\cite{BravermanOstrovsky2007}]\label{thm:SH}
Suppose that a function $f$ over intervals is $(\alpha, \beta)$-smooth and can be exactly computed using space $g$ and update time $h$ (with respect to addition at the end).
Then, there exists an algorithm that maintains a $(1 \pm \alpha)$-approximation of $f$ in the sliding window using $O(\beta^{-1}(g + \log n)\log n)$ bits of space and $O(\beta^{-1} h \log n)$ time per update.
\end{theorem}

\section{Main Result}\label{sec:main}
In this section, we prove Theorem~\ref{thm:main} with the aid of the Smooth Histogram (SH) framework.

\subsection{Basic Decomposition}
Let $f$ be the maximum subarray sum defined as \eqref{eq:MSS}.
In addition, for a sequence $W=(x_1, x_2, \ldots, x_n)$, we define the \emph{maximum prefix sum} and the \emph{maximum suffix sum} by
\[
\Pre(W) \coloneqq \max \left\{ 0,\max_{1 \le j \le n}\sum_{k=1}^{j}x_k \right\}
\]
and
\[
\Suf(W) \coloneqq \max \left\{ 0,\max_{1 \le i \le n}\sum_{k=i}^{n}x_k \right\},
\]
respectively.

The following observation is crucial in our analysis:
For any two adjacent intervals $A$ and $C$ in this order,
\begin{align}\label{eq:decomposition}
\hspace{-4mm}f(A \cup C)=\max\{f(A),\,f(C),\,\Suf(A)+\Pre(C)\}.
\end{align}
Indeed, an optimal subarray in $A \cup C$ must be contained entirely in $A$, entirely in $C$, or cross the boundary between them.
In the last case, it is the concatenation of a suffix of $A$ and a prefix of $C$.

\subsection{A Constant-Factor Approximation via Smoothness}
We first show that the maximum subarray sum satisfies a relaxed smoothness property, which is sufficient to obtain a constant-factor approximation by the standard SH framework.

\begin{lemma}\label{lem:constant_smooth}
For any constant $0<\beta<1$, the maximum subarray sum is $\left(\frac{1}{2-\beta},\beta\right)$-smooth.
\end{lemma}

\begin{proof}
First of all, $\beta \le \frac{1}{2 - \beta} < 1$ holds for every $0 < \beta < 1$, since $2 - \beta > 1$ and
\[\beta(2 - \beta) = -(1 - \beta)^2 + 1 \le 1.\]

We verify the conditions in Definition~\ref{def:smooth}.
The first condition immediately follows from the definition.
Monotonicity also holds: if $B \subseteqr A$, then every subarray of $B$ is also a subarray of $A$, and hence $f(B)\le f(A)$.

It remains to prove the stability under extension.
Let $B \subseteqr A$ be such that $(1-\beta)f(A)\le f(B)$, and let $C$ be any interval adjacent to both $A$ and $B$.
We show $\left(1 - \frac{1}{2 - \beta}\right) f(A \cup C) \le f(B \cup C)$ by considering three cases according to which term attains $f(A \cup C)$ in \eqref{eq:decomposition}.

\paragraph*{Case 1: $f(A \cup C)=f(C)$.}
Then
\[
f(B \cup C)\ge f(C)=f(A \cup C),
\]
so the claim is trivial.

\paragraph*{Case 2: $f(A \cup C)=f(A)$.}
Then
\[
f(B \cup C)\ge f(B)\ge (1-\beta)f(A) = (1 - \beta)f(A \cup C).
\]
Since $\beta \le \frac{1}{2-\beta}$, we obtain
\[
f(B \cup C)\ge \left(1-\frac{1}{2-\beta}\right)f(A \cup C).
\]

\paragraph*{Case 3: $f(A \cup C)=\Suf(A)+\Pre(C)$.}
Since $\Suf(A)\le f(A)$, the assumption implies
\[
f(B)\ge (1-\beta)f(A)\ge (1-\beta)\Suf(A).
\]
Since $\Pre(C) \le f(C)$ as well, we have
\begin{align*}
f(B \cup C) &\ge \max\{f(B), f(C)\}\\ &\ge \max\{(1-\beta)\Suf(A),\,\Pre(C)\}.
\end{align*}
For any nonnegative $x$ and $y$, by considering their convex combination, we see
\[
\max\{(1-\beta)x,y\}\ge \frac{1-\beta}{2-\beta}(x+y).
\]
Applying this with $x=\Suf(A)$ and $y=\Pre(C)$, we obtain
\begin{align*}
f(B \cup C) &\ge \frac{1-\beta}{2-\beta}\bigl(\Suf(A)+\Pre(C)\bigr)\\
&= \left(1-\frac{1}{2-\beta}\right)f(A \cup C).
\end{align*}

Therefore, in all cases,
\[
\left(1-\frac{1}{2-\beta}\right)f(A \cup C)\le f(B \cup C),
\]
which proves that $f$ is
$\left(\frac{1}{2-\beta},\beta\right)$-smooth.
\end{proof}

By Theorem~\ref{thm:SH} with Kadane's algorithm, this immediately yields a $(1 \pm \alpha)$-approximation using $O((\log n)^2)$ bits of space for a constant $\alpha > \frac{1}{2}$.
However, this is not sufficient for our goal of obtaining a $(1\pm\epsilon)$-approximation for any $\epsilon>0$.

\subsection{Toward $(1\pm\epsilon)$-Approximation}
The limitation of the standard SH approach is that it only uses the function values $f(A)$ to prune the maintained intervals.
This is insufficient when the optimal subarray crosses the boundary between two intervals, because in that case the relevant quantity is not $f(A)$ itself but the boundary term $\Suf(A)+\Pre(C)$.
In particular, it may happen that $f(B)$ is close to $f(A)$ while $\Suf(B)$ is much smaller than $\Suf(A)$.
Then replacing $A$ with its suffix $B$ may severely underestimate the value of $f(A \cup C)$ once a suitable interval $C$ is appended.

To overcome this difficulty, we check not only $f(A)$ but also $\Suf(A)$ when pruning.
For each interval $I_i = [s_i, t]$, Kadane's algorithm (cf.~Section~\ref{sec:Kadane}) actually maintains the pair $(\mu_{i,t}, \mu_{i,t}^*)$ that corresponds to $(\Suf(I_i), f(I_i))$.
We then modify the pruning rule of the SH framework described in Section~\ref{sec:SH} as follows.

Let $I_i \coloneqq [s_i, t]$ $(i = 1, 2, \dots, q)$ be the maintained instances.
In the original SH framework (with $\beta = \epsilon$), to maintain the last condition for the starting positions $s_i$, the algorithm checks for $i = 2, 3, \dots, q - 1$ in this order whether $f(I_{i+1}) \ge (1 - \epsilon)f(I_{i-1})$ or not and removes $I_i$ (and updates the indices accordingly) if yes.
In addition to this condition, we impose $\Suf(I_{i+1}) \ge (1 - \epsilon)\Suf(I_{i-1})$ for its removal.
Then, both the maximum subarray sum and the maximum suffix sum are preserved up to a factor $(1-\epsilon)$ after its removal.

The next lemma shows that this improved pruning rule is sufficient to preserve the contribution of subarrays crossing interval boundaries.

\begin{lemma}\label{lem:pruning}
Let $A$ and $B$ be intervals such that $B \subseteqr A$ and
\[
f(B)\ge (1-\epsilon)f(A),
\quad
\Suf(B)\ge (1-\epsilon)\Suf(A).
\]
Then, for any interval $C$ adjacent to both $A$ and $B$,
\[
f(B \cup C)\ge (1-\epsilon)f(A \cup C).
\]
\end{lemma}

\begin{proof}
Again, we use the decomposition in \eqref{eq:decomposition}, i.e.,
\begin{align*}
f(A \cup C)&=\max\{f(A),\,f(C),\,\Suf(A)+\Pre(C)\},\\
f(B \cup C)&=\max\{f(B),\,f(C),\,\Suf(B)+\Pre(C)\},
\end{align*}
and consider three cases of which term attains $f(A \cup C)$.

\paragraph*{Case 1: $f(A \cup C)=f(C)$.}
Then
\[
f(B \cup C)\ge f(C)=f(A \cup C).
\]

\paragraph*{Case 2: $f(A \cup C)=f(A)$.}
Then
\[
f(B \cup C)\ge f(B)\ge (1-\epsilon)f(A)
=(1-\epsilon)f(A \cup C).
\]

\paragraph*{Case 3: $f(A \cup C)=\Suf(A)+\Pre(C)$.}
Then
\begin{align*}
\Suf(B)+\Pre(C)&\ge (1-\epsilon)\Suf(A)+\Pre(C)\\
&\ge (1-\epsilon)(\Suf(A)+\Pre(C)),
\end{align*}
since $\Pre(C)\ge 0$.
Therefore,
\[
f(B \cup C)\ge \Suf(B)+\Pre(C)
\ge (1-\epsilon)f(A \cup C).
\]

Thus, in all cases, $f(B \cup C)\ge (1-\epsilon)f(A \cup C)$.
\end{proof}

This lemma plays the same role as the stability condition in the standard SH analysis.
Hence, by applying the same proof strategy as in Section~\ref{sec:SH} with the improved pruning rule, we obtain a $(1\pm\epsilon)$-approximation algorithm for the maximum subarray sum in the sliding window.

\subsection{Algorithm and Its Complexity}
An algorithm that incorporates the improved pruning rule is summarized as follows.

At each time $t$, it maintains a sequence of indices
\[
s_1 < s_2 < \cdots < s_q = t,
\]
where each $s_i$ corresponds to an instance $I_i = [s_i, t]$ for which the maximum subarray sum is computed by Kadane's algorithm.

When a new element arrives, the algorithm performs the following operations:
\begin{enumerate}
\item update all maintained instances by adding the new element at the end;
\item create a new instance for the singleton interval consisting of the new element;
\item remove redundant instances using the improved pruning rule based on both $f$ and $\Suf$, i.e., for $i = 2, 3, \dots, q - 1$ in this order, if $f(I_{i+1}) \ge (1 - \epsilon)f(I_{i-1})$ and $\Suf(I_{i+1}) \ge (1 - \epsilon)\Suf(I_{i-1})$, then remove $I_i$ (and update the indices accordingly);
\item delete all expired indices so that $p_{s_1}$ is expired and $p_{s_2}$ is active.
\end{enumerate}

As in the standard SH framework, the output is the value $f(I_2)$ maintained by the first active instance, which satisfies
\[(1 - \epsilon)f(W_t) \le f(I_2) \le f(W_t).\]

\begin{lemma}
For any $\epsilon>0$, the above algorithm maintains a $(1\pm\epsilon)$-approximation of the maximum subarray sum in the sliding window using $O(\epsilon^{-1}(\log n)^2)$ bits of space and $O(\epsilon^{-1} \log n)$ operations per update.
\end{lemma}

\begin{proof}
Each maintained instance stores a timestamp together with the two values $f(I_i)$ and $\Suf(I_i)$.
Since all elements are assumed to have absolute values $O(\mathrm{poly}(n))$, the algorithm uses $O(\log n)$ bits for each instance.

It remains to bound the number of maintained instances.
The improved pruning rule requires preserving both $f(I_i)$ and $\Suf(I_i)$.
Nevertheless, for any index $1 < i < q$, at least one of the following inequalities holds:
\begin{align}
f(I_{i+1}) &< (1 - \epsilon)f(I_{i-1}),\label{eq:dec_f}\\
\Suf(I_{i+1}) &< (1 - \epsilon)\Suf(I_{i-1}).\label{eq:dec_Suf}
\end{align}
Since $f$ is monotone and bounded by $O(\mathrm{poly}(n))$, \eqref{eq:dec_f} can hold $O(\epsilon^{-1} \log n)$ times.
Furthermore, $\Suf$ is also monotone with respect to adding elements at the beginning, which also bounds the number of occurrences of \eqref{eq:dec_Suf} by $O(\epsilon^{-1} \log n)$.
Thus, $q = O(\epsilon^{-1}\log n)$, which concludes that the algorithm uses $O(\epsilon^{-1}(\log n)^2)$ bits of space in total.

Finally, we discuss the computational time.
Each instance can be updated in $O(1)$ operations as with Kadane's algorithm, and there are $q = O(\epsilon^{-1}\log n)$ maintained instances.
Thus, Steps 1 and 2 require $O(\epsilon^{-1}\log n)$ operations.
In addition, Steps 3 and 4 can be easily implemented in $O(q)$ time, which completes the proof.
\end{proof}

\subsection{Extension to Nonempty Subarray Version}\label{sec:nonempty}
We here remark an extension to the (more standard) version of the maximum subarray sum problem that does not allow the empty subarray, i.e., the target statistics is defined for each sequence $W = (x_1, x_2, \dots, x_n)$ by
\[\tilde{f}(W) \coloneqq \max_{1 \le i \le j \le n} \sum_{k=i}^{j} x_k.\]
Note that $f(W) = \max\{0, \tilde{f}(W)\}$, and $\tilde{f}(W)$ is negative if and only if so are all $x_i$.
For the case where $\tilde{f}(W) < 0$, we naturally extend the definition of $(1 \pm \epsilon)$-approximation by flipping the inequality in \eqref{eq:approximation}, i.e., $\tilde{f}'(W)$ is a $(1 \pm \epsilon)$-approximation of $\tilde{f}(W)$ if
\[(1 - \epsilon)\tilde{f}(W_t) \ge \tilde{f}'(W_t) \ge (1 + \epsilon)\tilde{f}(W_t).\]
This is equivalent to requiring
\[|\tilde f'(W_t) - \tilde f(W_t)| \le \epsilon |\tilde f(W_t)|.\]

Under this definition, we can achieve the same result for this version of the problem as Theorem~\ref{thm:main}.
We only need to care for the case where $\tilde{f}(W) < 0$ in addition to the above discussion.
In that case, the optimum is simply the maximum element in the current window that only contains negative elements.
Thus, by sign inversion, this problem reduces to approximate maintenance of the minimum among positive elements, which is a fundamental problem.

To maintain an approximation of the minimum, a simple geometric discretization of the value domain suffices\footnote{This can also be regarded as an application of the SH framework, where the definition of $(\alpha, \beta)$-smoothness (Definition~\ref{def:smooth}) should be modified accordingly (since $\min$ is not monotonically increasing but \emph{decreasing}): for any pair of $A$ and $B$ with $B \subseteqr A$,
\begin{itemize}
    \item $f(B) \ge f(A)$, and
    \item if $(1 + \beta)f(A) \ge f(B)$, then $(1 + \alpha)f(A \cup C) \ge f(B \cup C)$ for any interval $C$ adjacent to $A$ and $B$.
\end{itemize}
Then, by an analogous discussion, one can maintain $(1 \pm \alpha)$-approximation of $f(W_t)$ using $O(\beta^{-1} \log n)$ instances.}:
\begin{itemize}
    \item divide the positive range $[1, M]$ into buckets of the form $[(1+\epsilon)^i,(1+\epsilon)^{i+1})$ $(i = 0, 1, \dots, r)$, where $M = O(\mathrm{poly}(n))$ is the maximum value in the stream, and 
    \item for each bucket, store the most recent occurrence time of an element in that bucket.
\end{itemize}
The smallest active nonempty bucket then gives a $(1\pm\epsilon)$-approximation to the minimum.
Clearly, the required space is $O(r\log n) = O(\epsilon^{-1}(\log n)^2)$ bits, and its maintenance is done by $O(r) = O(\epsilon^{-1} \log n)$ operations per update.

Thus, by combining this simple procedure for the all-negative case with the refined SH framework for the case $\tilde f(W) \ge 0$, one obtains the same asymptotic bounds for the nonempty subarray version of the problem as in Theorem~\ref{thm:main}.

\section{Concluding Remarks}\label{sec:conclusion}

In this paper, we have studied the problem of maintaining the maximum subarray sum in the sliding window model.
Since the maximum subarray sum is not simply decomposed into intervals, the standard Exponential Histogram technique looks difficult to apply directly to this problem.
Our main contribution is to show that this problem can nevertheless be handled within a refined Smooth Histogram framework.
We have first proved that the maximum subarray sum satisfies a relaxed smoothness property, which yields a constant-factor approximation.
We have then improved the pruning rule using maximum suffix sums in addition, which enables us for any $\epsilon > 0$ to maintain a $(1\pm\epsilon)$-approximation using $O(\epsilon^{-1}(\log n)^2)$ bits of space and $O(\epsilon^{-1}\log n)$ operations per update.
The space complexity is asymptotically optimal.

Future work includes clarifying the status of the update time.
Indeed, for basic problems including counting 1s, it can be improved to $O(1)$ in the amortized sense.
The present algorithm explicitly updates all maintained instances, but improving the current $O(\epsilon^{-1}\log n)$ bound would likely require a more implicit representation of the maintained information.
Another interesting direction is to investigate whether the same suffix-based idea can be applied to other statistics that are neither weakly additive nor directly amenable to the standard smooth histogram framework in the sliding window model.

\section*{Acknowledgments}
This work was supported by JSPS KAKENHI Grant Number JP25H01114 and JST CRONOS Japan Grant Number JPMJCS24K2.

\bibliographystyle{plain}
\bibliography{ref}

\end{document}